\documentclass[10pt,leqno,letterpaper]{article}
\usepackage[margin=1in]{geometry}
\usepackage{enumerate}
\usepackage{rpmacros}
\usepackage[T1]{fontenc}
\usepackage{gitinfo}
\usepackage[usenames,dvipsnames]{color}
\usepackage{stmaryrd}
\usepackage{xfrac}
\usepackage{mathpazo}
\usepackage[normalem]{ulem}
\usepackage{bm}
\usepackage{todonotes}
\usepackage{lipsum}
\usepackage{setspace}
\usepackage{scrextend}
\newcommand*\samethanks[1][\value{footnote}]{\footnotemark[#1]}

\usetikzlibrary{decorations.pathreplacing,arrows}
\RequirePackage[colorlinks=true]{hyperref}
\hypersetup{
  linkcolor=[rgb]{0.3,0.3,0.6},
  citecolor=[rgb]{0.2, 0.6, 0.2},
  urlcolor=[rgb]{0.6, 0.2, 0.2}
}

\usepackage[capitalise,nameinlink]{cleveref}

\usepackage{amsthm}
\usepackage{thmtools,thm-restate}

\numberwithin{equation}{section}
\declaretheoremstyle[bodyfont=\it,qed=\qedsymbol]{noproofstyle}

\declaretheorem[numberlike=equation]{observation}
\declaretheorem[numberlike=equation,style=noproofstyle,name=Observation]{observationwp}
\declaretheorem[name=Observation,numbered=no]{observation*}
\crefname{observation}{Observation}{Observations}
\crefname{observationwp}{Observation}{Observations}

\declaretheorem[numberlike=equation]{theorem}

\declaretheorem[name=Theorem,numbered=no]{theorem*}

\declaretheorem[numberlike=equation]{lemma}
\declaretheorem[name=Lemma,numbered=no]{lemma*}

\declaretheorem[numberlike=equation]{corollary}
\declaretheorem[name=Corollary,numbered=no]{corollary*}
\declaretheorem[numberlike=equation,style=noproofstyle,name=Corollary]{corollarywp}

\declaretheorem[name=Proposition,numbered=no]{proposition*}

\declaretheorem[name=Claim,numbered=no]{claim*}

\declaretheorem[name=Conjecture,numbered=no]{conjecture*}

\declaretheorem[name=Question,numbered=no]{question*}

\declaretheoremstyle[bodyfont=\it,qed=$\lozenge$]{defstyle} 

\declaretheorem[numberlike=equation,style=defstyle]{definition}
\declaretheorem[unnumbered,name=Definition,style=defstyle]{definition*}

\declaretheorem[unnumbered,name=Example,style=defstyle]{example*}

\declaretheorem[unnumbered,name=Notation=defstyle]{notation*}

\declaretheorem[unnumbered,name=Construction,style=defstyle]{construction*}

\declaretheorem[unnumbered,name=Remark,style=defstyle]{remark*}

\newtheorem{dtheorem}[theorem]{Descent Theorem}
\declaretheorem[numberlike=equation,style=noproofstyle,name=Courant-Fischer Theorem]{cftheorem}

\usepackage{nth}
\usepackage{intcalc}
\usepackage{etoolbox}
\usepackage{xstring}

\usepackage{ifpdf}
\ifpdf
\else
\usepackage[quadpoints=false]{hypdvips}
\fi

\newcommand{\ECCCurl}[2]{http://eccc.hpi-web.de/report/\ifnumcomp{#1}{>}{93}{19}{20}#1/#2/}

\newcommand{\shortECCC}[2]{\texttt{\href{http://eccc.hpi-web.de/report/\ifnumcomp{#1}{>}{93}{19}{20}#1/#2/}{eccc:TR#1-#2}}}

\newcommand{\parseECCC}[1]{
\StrSubstitute{#1}{TR}{}[\tmpstring]%
\IfSubStr{\tmpstring}{/}{ 
\StrBefore{\tmpstring}{/}[\ecccyear]%
\StrBehind{\tmpstring}{/}[\ecccreport]%
}{
\StrBefore{\tmpstring}{-}[\ecccyear]%
\StrBehind{\tmpstring}{-}[\ecccreport]%
}%
\shortECCC{\ecccyear}{\ecccreport}}

\newif\ifnote
\notetrue
\ifnote
\newcommand{\RPnote}[1]{\todo[color=red!80!green!25!blue!33, size=\footnotesize]{RP: #1}}
\newcommand{\phnote}[1]{\todo[color=red!100!green!33,
  size=\footnotesize]{ph: #1}}
\newcommand{\prahladhuvacha}[1]{\todo[color=red!100!green!33,inline,size=\small]{Sri
    Prahladh Uvacha: #1}}
\else
\newcommand{\RPnote}[1]{}
\newcommand{\phnote}[1]{}
\newcommand{\prahladhuvacha}[1]{}
\fi

\newcommand{\X}{\mathcal{X}}

\onehalfspace
\allowdisplaybreaks
\title{A note on the elementary construction of\\
  High-Dimensional Expanders\\
  of Kaufman and Oppenheim\thanks{ This work was done when the
    authors were visiting the Simons Institute for Theory of
    Computing, Berkeley to participate in the summer cluster on
    Error-Correcting Codes and High Dimensional Expanders 2019..}}
\author{Prahladh Harsha\thanks{Tata Institute of Fundamental Research,
    INDIA. email: {\tt
      \{prahladh,ramprasad\}@tifr.res.in}. Research of
    the authors supported by the Department of Atomic Energy,
    Government of India, under project no. 12-R\&D-TFR-5.01-0500. Research
    of the first and second authors supported in part by the Swarnajayanti
    and Ramanujan Fellowships respectively.}
  \and
  Ramprasad Saptharishi\samethanks }
\begin{document}
\maketitle
{\let\thefootnote\relax
\footnotetext{\textcolor{white}{Base version:~(\gitAuthorIsoDate)\;,\;\gitAbbrevHash\;\; \gitVtag}}
}

\begin{abstract}
  In this note, we give a self-contained and elementary proof of the elementary construction of spectral high-dimensional expanders using elementary matrices due to Kaufman and Oppenheim [{\em Proc.\ $50$th ACM Symp.\ on Theory of Computing (STOC)}, 2018].
As a bonus, this also yields a simple construction and analysis of standard expanders. 
\end{abstract}

\section{Introduction}

In the last few years, there has been a surge of activity related to high-dimensional expanders (HDXs).
Loosely speaking, high-dimensional expanders are a high-dimensional generalization of classical graph expanders.
Depending on which definition of graph expansion is generalized, there are several different (and unfortunately, many a time mutually inequivalent) definitions of HDXs.
For the purpose of this note, we will restrict ourselves to the spectral definition of HDXs (see \cref{def:hdx}).
Lubotzky, Samuels and Vishne~\cite{LubotzkySV2005-exphdx,LubotzkySV2005-hdx} constructed high-dimensional analgoues of the Ramanujan expanders of Lubotzky, Philips and Sarnak~\cite{LubotzkyPS1988}, which they termed Ramanujan complexes.
These Ramanujan complexes have several desirable properties and gave rise to the first construction of constant degree spectral HDXs.
The Ramanujan graphs have the nice property that they are simple to describe, however their proof of expansion is extremely involved.
The Ramanujan complexes, on the other hand, are both non-trivial to describe as well as to prove their high-dimensional expansion property.
Subsequently Kaufman and Oppenheim~\cite{KaufmanO2018} gave an extremely elegant and elementary construction of spectral HDXs using elementary matrices.
Despite their construction being elementary and simple, the proof of expansion, though straightforward, requires some knowledge of some representation theory of the specific groups involved in the construction.
The purpose of this exposition is to give an alternate elementary proof of the expansion of the Kaufman-Oppenheim HDX construction.

The underlying graph of a HDX (even a one-sided-spectral HDX) is a two-sided-spectral expander.
Thus, this construction has the added advantage that it yields an elementary construction (accompanied with a simple proof) of a standard two-sided-spectral expander (though not an optimal one).

\section{Preliminaries}

We begin by recalling what a simplicial complex is. 
\begin{definition}[Simplicial complex]
  A \emph{simplicial complex} $X$ over a finite set $U$ is a collection of subsets of $U$ with the property that if $S \in X$ then any $T \subseteq S$ is also in $X$.
  \begin{itemize}
  \item For all $i\geq -1$,
    define $X(i) := \setdef{S \in X}{|S| = i+1}$. Thus, if $X$ is
    non-empty, then $X(-1) = \set{\emptyset}$. 
  \item The elements of $X$ are called \emph{simplices} or
    \emph{faces}. The elements of $X(0)$, $X(1)$ and $X(2)$ are
    usually referred to as \emph{vertices}, \emph{edges} and
    \emph{triangles} respectively.
  \item The graph defined by $X(0)$ and $X(1)$ is called the
    \emph{$1$-skeleton} of the complex. More generally, for any $1
    \leq k \leq d$, the $k$-skeleton of the complex $X$ is the
    sub-complex $X(-1) \cup X(0) \cup X(1) \cup \dots \cup X(k)$. 
  \item The \emph{dimension} of the simplicial complex $X$ is defined as the largest $d$ such that $X(d)$ (which consists of faces of size $d+1$) is non-empty.
  \item The simplicial complex is said to be \emph{pure} if every face
    is contained in some face in $X(d)$, where $d = \dim(X)$.
    
  \item For a face $S \in X$, the \emph{link} of $S$, denoted by
    $X_S$, is the simplicial complex defined as 
    \[
      X_S := \setdef{T \setminus S}{T\in X\;,\; S \subseteq
        T}. \qedhere\
    \]
    
  \end{itemize}
\end{definition}

Thus, a graph $G=(V,E)$ is just a simplicial complex $G$ of dimension one
with $G(0) = V$ and $G(1)=E$. We will deal with {\em weighted} pure
simplicial complexes where the weight function satisfies a certain
{\em balance} condition.

\begin{definition}[weighted pure simplicial complexes]
Given a $d$-dimensional pure simplicial complex $X$ and an associated
weight function $w\colon X \to \R_{\geq 0}$, we say the weight
function is {\em balanced} if the following two conditions are
satisfied.
\begin{align}\label{eq:weight}
  \sum_{\sigma \in X(d)} w(\sigma) & = 1 \; ; & 
  w(\sigma) &  = \frac{1}{i+2}\sum_{\tau \in X(i+1), \tau \supset \sigma}
              w(\tau), \ \text{ for all } i<d \text{ and } \sigma \in X(i).
\end{align}
A {\em weighted simplicial complex} $(X,w)$ is a pure simplical
      complex accompanied with a balanced weight function $w$. If no
      weight function is specified, then we work with the balanced weight
      function $w$ induced by the uniform distribution on the set
      $X(d)$ of maximal faces.

      For a face $S \in X$, the balanced weight function $w_S$ associated
      with the link $X_S$ is the restricted weight function, suitably
      normalized, more precisely $w_S :=
    \sfrac{w|_{X_S}}{w(S)}$. 
\end{definition}

Condition~\eqref{eq:weight} states that the weight function can be interpreted as a family of joint distributions $(w|_{X(-1)}, \ldots, w|_{X(d)})$ where $w|_{X(i)}$ is a probability distribution on $X(i)$.
The distribution $w|_{X(d)}$ is specified by the first condition in \eqref{eq:weight} while the second condition implies that the weight distribution $w|_{X(i)}$ is the distribution on $X(i)$ obtained by picking a random $\tau \in X(d)$ according to $w|_{X(d)}$ and then removing $(d-i)$ elements uniformly at random.

We now recall the classical definition
of what it means for a graph to be a spectral expander.

\begin{definition}[spectral expander]
Given an undirected  weighted graph $G=(V,E, w)$ on $n$ vertices, let $A_G$ be its
normalized adjacency matrix given as follows:
\[ A_G(u,v) := \begin{cases}
        \frac{w(u,v)}{w(u)} & \text{if } \{u,v\} \in E,\\
      0 & \text{otherwise.}
    \end{cases}
    \]
Let $1 = \lambda_1 \geq \lambda_2 \geq
\cdots \geq \lambda_n \geq -1$ be the $n$ eigenvalues of $A_G$ with
multiplicities in non-increasing order\footnote{By the balance
  condition, $w$ satisfies $w(v) = \sum_{\{u,v\} \in E} w(u,v)$. The
  matrix $A_G$ is self-adjoint with respect to the
  inner product $\inangle{f,g}_w := \E_{v \sim w}[f(v) g(v)]$ since
  $\inangle{f,Ag}_w = \inangle{Af,g}_w = \E_{\{u,v\} \sim w}[f(u)
  g(v)]$. Hence, $A_G$ has $n$ real eigenvalues which can be
  obtained using the \autoref{thm:cf}.}. We denote the second largest
eigenvalue of $G$ as $\lambda(G)$. 

$G$ is said to be a
\emph{$\lambda$-spectral expander} if $\max\{\lambda_2,|\lambda_n|\}
\leq \lambda$. This is sometimes also referred to as a $\lambda$-two-sided-spectral expander. 

$G$ is said to be a
\emph{$\lambda$-one-sided-spectral expander} if $\lambda_2 \leq \lambda$.
\end{definition}

This spectral definition of expanders is generalized to higher dimensional
simplicial complexes as follows.

\begin{definition}[$\lambda$-spectral HDX]\label{def:hdx}
  A weighted simplicial complex $(X,w)$ of dimension $d \geq 1$ is said to be a
  \emph{$\lambda$-spectral HDX} (or a
  \emph{$\lambda$-two-sided-spectral HDX})\footnote{These are sometimes also referred
    to as $\lambda$-link HDXs or $\lambda$-local-expanders to
    distinguish from an alternative global definition of
    high-dimensional expansion.} if for every $-1\leq i \leq d-2$ and
  $s \in X(i)$, the weighted $1$-skeleton of the link $(X_s, w_s)$ is a
  $\lambda$-spectral expander.

  A weighted simplicial complex $(X,w)$ of dimension $d \geq 1$ is said to be a
  \emph{$\lambda$-one-sided-spectral HDX} if for every $-1\leq i \leq d-2$ and
  $s \in X(i)$, the weighted $1$-skeleton of the link $(X_s,w_s)$ is a
  $\lambda$-one-sided-spectral expander.
\end{definition}

Using Garland's technique~\cite{Garland1973}, Oppenheim~\cite{Oppenheim2018} showed that if the 1-skeletons of all the links are
connected, then a spectral gap at dimension $(d-2)$ descends to all
lower levels.

\begin{dtheorem}[\cite{Oppenheim2018}]\label{thm:trickle-down} Suppose
  $(X,w)$ is a $d$-dimensional weighted simplicial complex with the following properties.
  \begin{itemize}
  \item For all $s \in X(d-2)$, the link $(X_s,w_s)$ is a
    $\lambda$-one-sided-spectral expander for some $\lambda <
    \frac{1}{d-1}$. 
  \item The $1$-skeleton of every link is connected.
  \end{itemize}
  Then, $(X,w)$ is a $\pfrac{\lambda}{1 - (d-1)\lambda}$-one-sided-spectral HDX.   
\end{dtheorem}
Thus to prove that a given simplicial complex is a spectral HDX, it suffices to show that the 1-skeleton of all links are connected and a spectral gap at the top level.
For the sake of completeness, we give a proof of the \autoref{thm:trickle-down} in \cref{sec:descent} which includes a descent theorem for the least eigenvalue as well.

\section{Coset complexes}

The HDX construction of Kaufmann and Oppenheim is a particular instantiation of a certain type of simplicial complex called a \emph{coset complex} based on a group and its subgroups. In this section, we give an exposition of these objects. 
For a basic primer on group theory,
see \cref{sec:group}.

\begin{definition}[coset complex] Let $G$ be a group and let
  $K_1,\ldots, K_d$ be $d$ subgroups of $G$. The
coset complex $\X(G,\set{K_1,\ldots, K_d})$ is a $(d-1)$-dimensional
simplicial complex defined as follows:
\begin{itemize}
\item The \emph{vertices}, $\X(0)$, consist of cosets of $K_1,\ldots, K_d$ and we shall say cosets of $K_i$ are of \emph{type} $i$.
\item The \emph{maximal faces}, $\X({d-1})$, consist of $d$-sets of
  cosets of different types with a non-empty intersection. That is,
  \[
    \set{g_1 K_1,\ldots, g_d K_d}\in \X(d-1) \Longleftrightarrow g_1 K_1 \cap \cdots \cap g_d K_d \neq \emptyset.
  \]
  An equivalent way of stating this is that $\set{g_1 K_1,\ldots, g_d K_d}\in \X(d-1)$ if and only if there is some $g \in G$ such that $g_i K_i = g K_i$ for all $i$, since
  \[
    g_i K_i = g K_i \;\Longleftrightarrow\; K_i = g_i^{-1} g K_i  \;\Longleftrightarrow\; g_{i}^{-1} g \in K_i \;\Longleftrightarrow\; g \in g_{i} K_i.
  \]
\item The lower dimensional faces are obtained by \emph{down-closing}
  the maximal faces. Hence, for $0\leq r \leq d$, $\set{g_{i_1} K_{i_1}, \ldots, g_{i_r} K_{i_r}} \in \X(r-1)$ if and only if $i_j \neq i_k$ for all $j \neq k$ and
  \[
    g_{i_1} K_{i_1} \cap \cdots \cap g_{i_r} K_{i_r} \neq \emptyset.
  \]
  We shall call the set $\set{i_1,\ldots, i_r}$ the \emph{type} of this face.
\item The dimension of this complex is $d-1$.
\item The weight function we will use is the one induced by the
  uniform distribution on the set $\X(d-1)$ of maximal faces. 
  \qedhere
\end{itemize}
\end{definition}

A simplicial complex constructed this way is \emph{partite} in the sense that each maximal face consists of vertices of distinct types. 

It follows from the definition, that $\X(i)$ is precisely the set of
cosets of the form $g K_S$ where $K_S = \intersection_{j\in S} K_j$
for sets $S \subseteq [d]$ of size exactly $i+1$. In particular, $\X(d-1)$, the
set of maximal faces, is in 1-1 correspondence with the group
$G$ if $\intersection_{j\in [d]} K_j = \{\mathrm{id}\}$ where
``$\mathrm{id}$'' is the identity element of the group $G$.

\subsubsection*{Connectivity:}

\begin{observation}\label{obs:nonempty-coset-intersection}
  $g_1 K_1 \cap g_2 K_2 \neq \emptyset$ if and only if $g_1^{-1} g_2 \in K_1 K_2$. 
\end{observation}
\begin{proof}
  ($\Rightarrow$) Say $x = g_1 k_1 = g_2 k_2$ for $k_1 \in K_1$ and $k_2 \in K_2$. Then $g_1^{-1} g_2 = g_1^{-1} x \cdot x^{-1} g_2 = k_1 k_2^{-1} \in K_1 K_2$.

  ($\Leftarrow$) If $g_1^{-1} g_2 = k_1 k_2$ for $k_1 \in K_1$ and $k_2 \in K_2$, then $g_1 k_1 = g_2 k_2^{-1} \in g_1 K_1 \cap g_2 K_2$. 
\end{proof}

\begin{lemma}[Criterion for connected $1$-skeletons]\label{lem:conn-criterion}
  The $1$-skeleton (underlying graph) of $\X(G,\set{K_1,\ldots, K_d})$ is connected if and only if $G = \inangle{K_1,\ldots, K_d}$. 
\end{lemma}
\begin{proof}
  ($\Leftarrow$) Since there is always an edge between $gK_i$ and $gK_j$ for $i\neq j$, it suffices to show that $K_1$ is connected to $gK_1$ for an arbitrary $g \in G$.
Suppose, for an arbitrary element $g \in G$, we have $g = g_1 \ldots
g_r$ where $g_j \in K_{i_j}$ and $i_j \neq i_{j+1}$ for each $j$. We might, without loss of generality, assume that
(a) $g_1 \in K_1$ (otherwise set $g = 1\cdot g_1 \cdots g_r$)  and (b) if $r \geq 2$,
then $i_r \neq 1$ (since otherwise we might then have worked with $g' =
g_1g_2 \ldots g_{r-1}$ as $gK_1= g'g_r K_1 = g'K_1$).

Then, we get the following path connecting $K_1 $ and $g K_{i_r}$
  \[
    K_1 =  g_1 K_{i_1} \rightarrow (g_1 g_2) K_{i_2} \rightarrow (g_1 g_2 g_3) K_{i_3} \rightarrow \ldots \rightarrow (g_1\cdots g_r) K_{i_r} = g K_{i_r}.
  \]
  Note that, due to \cref{obs:nonempty-coset-intersection}, each
  successive pair of cosets are connected by an edge in the simplicial
  complex. Now, since $gK_{i_r}$ is adjacent to $gK_1$ (as $i_r \neq
  1$), we have that $K_1$ is connected to $gK_1$. \\

  \noindent
  ($\Rightarrow$) For an arbitrary $g \in G$, since the $1$-skeleton is connected we have a path
  \[
    K_1 = g_0 K_{i_0} \rightarrow  g_1 K_{i_1} \rightarrow \cdots \rightarrow g_r K_{i_r} = g K_1.
  \]
  By \cref{obs:nonempty-coset-intersection}, for every $j = 0,\ldots, r-1$, we have $g_j^{-1} g_{j+1} \in K_{i_j} K_{i_{j+1}} \in \inangle{K_1,\ldots, K_d}$.
  Therefore,
  \[
    g = (g_0^{-1} g_1) \cdot (g_1^{-1} g_2) \cdots (g_{r-1}^{-1} g_r) \in \inangle{K_1,\ldots, K_d}. \qedhere
  \]  
\end{proof}

\subsubsection*{Structure of links of the coset complex:}

For any set $S \subseteq [d]$, define the group $K_S :=
\Intersection_{i\in S} K_i$; let $K_\emptyset := \inangle{K_1,\ldots,
  K_d}$. The following lemma shows that the links of a coset complex
are themselves coset complexes. 

\begin{lemma}\label{lem:link-structure}
  For any $v\in \X(k)$ of type $S \subseteq [d]$, the link $X_v$ is isomorphic to the simplicial complex defined by $\X(K_S, \setdef{K_{S} \cap K_i}{i\notin S})$. 
\end{lemma}
\begin{proof}
  It suffices to prove this lemma for $v \in \X(0)$ as links of higher levels can be obtained by inductive applications of this case.

Observe that if $g$ is any element of  $G$, then
  $(g_{i_1} K_{i_1}, \ldots, g_{i_r} K_{i_r}) \in \X(r-1)$ if and only
  if $(g g_{i_1} K_{i_1}, \ldots, g g_{i_r} K_{i_r}) \in \X(r-1)$.
  Therefore, the link of a coset $g K_i$ is isomorphic to the link of
  the coset $K_i$. Thus, it suffices to prove the lemma for links of
  the type $X_{K_i}$ for some $i \in [d]$. 

  Let $v$ be the coset $K_1$, without loss of generality.
  The \emph{vertices} of the link, $X_v(0)$, are cosets of $K_2,\ldots, K_d$ that have a non-empty intersection with $K_1$.
  Note that any non-empty intersection $g_j K_j \cap K_1$ of a coset with $K_1$ is itself a coset $\tilde{g}_j (K_j \cap K_1)$ of the intersection subgroup $K_j \cap K_1$ in $K_1$. Indeed, suppose that $g_j h_j \in K_1$ for some $h_j \in K_j$. Then, $g_j h_j K_j = g_j K_j$ and $g_j h_j K_1 = K_1$ and hence
  \[
    g_j K_j \cap K_1 = g_j h_j K_j \cap g_j h_j K_1 = g_j h_j (K_j \cap K_1).
  \]
  Therefore, the vertices of the link $X_v(0)$ are in bijective correspondence with cosets of $\setdef{K_j \cap K_1}{j \in \set{2,\ldots, d}}$.

  The maximal faces in $X$ that contain the coset $K_1$ are precisely
  $d$-sets of cosets $\set{K_1, g_2 K_2, \ldots, g_d K_d}$ with a
  non-empty intersection and hence
  \[
    \emptyset \neq K_1 \cap g_2 K_2 \cap  \cdots \cap g_d K_d = (g_2 K_2 \cap K_1) \cap \cdots \cap (g_d K_d \cap K_1) = \tilde{g}_2 (K_2 \cap K_1) \cap \cdots \cap \tilde{g}_d (K_d \cap K_1),
  \]
  which are  precisely the maximal faces of the coset complex
  $\X\inparen{K_1,\setdef{K_j \cap K_1}{j \in \set{2,\ldots,
        d}}}$.  This establishes the isomorphism between $X_v$ and $\X\inparen{K_1,\setdef{K_j \cap K_1}{j \in \set{2,\ldots,
        d}}}$.
\end{proof}

\section{A concrete instantiation}

The simplicial complex of Kaufman and Oppenheim~\cite{KaufmanO2018} is
a specific instantiation of the above \emph{coset complex}
construction. This section is devoted to an exposition of this instantiation of Kaufman and Oppenheim.  We will need some notation to describe their group.

\subsection*{Notation}

\begin{itemize}
\item Let $R$ denote the ring $\frac{\F_p[t]}{\inangle{t^s}}$. This is
  a ring whose elements can be identified with polynomials in
  $\F_p[t]$ of degree less than $s$ (where addition and multiplication
  are performed modulo $t^s$). We will think of $p$ as some fixed
  prime power, $t$ as a formal variable and $s$ as a growing integer.
\item For any $d\geq 3$, and $1\leq i,j\leq d$ with $i \neq j$ and an element $r \in R$,
  we define  $e_{i,j}(r)$ to be the $d\times d$ elementary matrix with $1$'s on the diagonal and $r$ on the $(i,j)$-th entry.

  For the sake of notational convenience, we shall abuse this notation
  and use $e_{i+d,j}(r), e_{i,j+d}(r)$ etc. to refer to $e_{i,j}(r)$ by
  \emph{wrapping around} if necessary. For example, $e_{d,d+1}(r)$
  refers to $e_{d,1}(r)$.
\end{itemize}

\noindent
We are now ready to describe the groups in the construction.
\begin{align*}
  \text{For $i\in \set{1,\ldots, d}$,}\quad K_i & = \inangle{e_{j,j+1}(at + b)\;:\; a,b \in \F_p\,,\,j \in [d] \setminus \set{i}}.\\
  G & = \inangle{K_1,\ldots, K_d}
\end{align*}
Each $K_i$ is generated by elementary matrices that have $1$'s on the
diagonal and an arbitrary linear polynomial in one entry of the
generalised diagonal $\set{(i,j)\;:\; i+1= j \bmod d}$.

It so happens that the group $G$ generated by the subgroups
$K_1, \ldots, K_d$ is $\mathrm{SL}_d(R)$, the group of $d\times d$
matrices with entries in $R$ whose determinant is $1$ (in $R$). This
is a non-trivial fact. All we will need is the simpler fact that $|G|$
grows exponentially with $s$ (for fixed $p$ and $d$) while the size of
the groups $K_i$ are functions of $p$ and $d$ (and independent of
$s$). This will follow from the sequence of observations and lemmas
developed in the following section.\\

\noindent Given the above definition, there are two ``different'' subgroups we can define. 
\begin{align*}
K_S & = \intersection_{i\in S} K_i,\\
\widetilde{K_S} &:= \inangle{e_{i,i+1}(at+b)\;:\; a,b\in \F_p\,,\, i\notin S}.
\end{align*}
That is, $K_S$ is the intersection of the groups $\setdef{K_i}{i\in S}$, and $\widetilde{K_S}$ is the group generated by the intersection of the generators of the $K_i$'s. Thus, clearly, $\widetilde{K_S} \subseteq K_S$. The following lemma shows that in fact the two groups are identical.

\begin{lemma}[Intersections of $K_i$'s]\label{cor:intersection-property-KS}
  For any $S \subseteq [d]$,
  \[
   \widetilde{K_{S}} = \inangle{e_{i,i+1}(at +b)\;:\; a,b\in \F_p\,,\, i\notin S}
    = \Intersection_{i\in S} K_i = K_S
  \]
  In other words, the group generated by the intersection of generators equals the group intersection. 
\end{lemma}

We will prove this lemma in the following section by giving an
explicit description of the groups that makes the above lemma
evident. An immediate consequence of this lemma is that $K_{[d]} =
\{\mathrm{id}\}$ and hence $\X(d-1)$, the set of maximal faces, is in
1-1 correspondence with the group $G$.

\subsection{Explicit description of the groups}

The following is an easy consequence of the definition of
$e_{i,j}(r)$. Note that $e_{i,j}(r)$ is defined only if $i \neq j$. 

\begin{observationwp}\label{obs:comm-generators}

  \begin{enumerate}[(a)]
    \item\label{item:comm-gen-sum} \emph{Sum:} $e_{i,j}(r_1) \cdot e_{i,j}(r_2) =
      e_{i,j}(r_1 + r_2)$.

      As a corollary, $e_{i,j}(r)^{-1} = e_{i,j}(-r)$. 
     \item\label{item:comm-gen-product} \emph{Product:} If $i \neq \ell$, the commutator\footnote{The commutator of two elements $g,h$, denoted by $[g,h]$ is defined as $g^{-1} h^{-1} g h$. (\autoref{defn:commutator})}
       $[e_{i,j}(r_1),e_{k,\ell}(r_2)]$ behaves as follows.
  \[
    [e_{i,j}(r_1),e_{k,\ell}(r_2)] = 
    \begin{cases}
      e_{i,\ell}(r_1r_2) & \text{if $j = k$},\\
      \mathrm{id} & \text{if $j \neq k$}.
    \end{cases}\qedhere
  \]
  \end{enumerate}
\end{observationwp}
\begin{proof}
  Let $\mu_{i,j}$ denote the matrix that has a $1$ at the $(i,j)$-th entry, and $0$ everywhere. Then, \eqref{item:comm-gen-sum} follows as
  \begin{align*}
    e_{i,j}(r_1) e_{i,j}(r_2) &= (I + \mu_{i,j} r_1) \cdot (I + \mu_{i,j}(r_2))\\
    & = I + \mu_{i,j} \cdot (r_1 + r_2)\qquad\text{(since $\mu_{i,j}^2 = 0$ when $i\neq j$)}.
  \end{align*}

  As for \eqref{item:comm-gen-product}, we follow along a similar calculation. Note that 
  \[
  \mu_{i,j} \cdot \mu_{k,\ell} = \begin{cases} 0 & \text{if $j\neq k$}\\ \mu_{i,\ell} & \text{if $j = k$}.\end{cases}
  \]
  Therefore, 
  \begin{align*}
    [e_{i,j}(r_1), e_{k,\ell}(r_2)] & = (I - \mu_{i,j} r_1)\cdot (I - \mu_{k,\ell}r_2) \cdot (I + \mu_{i,j} r_1)\cdot (I + \mu_{k,\ell}r_2).
  \end{align*}
  When $j = k$ (along with the assumption that $i\neq j$, $k\neq \ell$), this simplifies to
  \begin{align*}
    [e_{i,j}(r_1), e_{k,\ell}(r_2)] &= (I - \mu_{i,j}r_1 - \mu_{k,\ell}r_2 + \mu_{i,\ell} \cdot r_1r_2) \cdot (I + \mu_{i,j}r_1 + \mu_{k,\ell}r_2 + \mu_{i,\ell} \cdot r_1r_2)\\
    & = I + \mu_{i,j}(r_1 - r_1) + \mu_{k,\ell}(r_2 - r_2) + \mu_{i,\ell}(r_1r_2 + r_1r_2 - r_1r_2)\\
    & = I + r_1r_2 \mu_{i,\ell}.
  \end{align*}
  If $j \neq k$, then we get
  \begin{align*}
  [e_{i,j}(r_1), e_{k,\ell}(r_2)] & = (I - \mu_{i,j} r_1)\cdot (I - \mu_{k,\ell}r_2) \cdot (I + \mu_{i,j} r_1)\cdot (I + \mu_{k,\ell}r_2).\\
  & = I + \mu_{i,j}(r_1 - r_1) + \mu_{k,\ell}(r_2 - r_2) = I\qedhere
  \end{align*}
\end{proof}

\noindent 
Therefore, for distinct $i,j,k \in [d]$ (which exist when $d\geq 3$), we have
\[
  \insquare{e_{i,j}(r_1), \insquare{e_{j,i}(r_2), e_{i,j}(r_3)}} = e_{i,j}(r_1r_2r_3)
\]
Thus, for all $d \geq 3$, using the above observation along with \autoref{obs:comm-generators}\eqref{item:comm-gen-sum}, we get that $e_{i,j}(r)$ for any $r\in R$ can be generated by $\setdef{e_{k,\ell}(at + b)}{k,\ell \in [d]\;,\; a,b\in \F_p\;}$. 
This in particular implies that $|G|$ is at least $p^s$.
On the other hand, the size of $K_i$ depends only on $d,p$ and is independent of $s$.
The lemma below describes $K_d$; the other $K_i$'s are just rearrangements of rows and columns in $K_d$.

\begin{lemma}[Explicit description of $K_d$]
  The group $K_d = \inangle{e_{i,i+1}(at + b)\;:\; a,b\in \F_p\,,\, i\neq d}$ consists of matrices $A = (A_{i,j})$ of the following form:
  \begin{align*}
    A_{i,j} = \begin{cases}
      1 & \text{if }i = j,\\
      \text{a polynomial of degree $ \leq r$} & \text{if }j - i = r,\\
      0 & \text{if }i > j.
    \end{cases}
  \end{align*}
  \noindent 
  More generally, stating the above differently, for any $n \in [d]$, the group $K_n = \inangle{e_{i,i+1}(at + b)\;:\; a,b\in \F_p\,,\, i\neq n}$ consists of matrices $A = (A_{i,j})$ of the following form:
  \begin{align*}
    A_{i,j} = \begin{cases}
      1 & \text{if }i = j,\\
      \text{a polynomial of degree $(j - i)\bmod{d}$} & \text{if }j\neq i \text{ and }n\not\in\set{i,i+1,\ldots,j-1}_{\bmod{d}},\\
      0 & \text{otherwise}.
    \end{cases}
  \end{align*}
\end{lemma}
\begin{proof}
Follows easily from repeated applications of \cref{obs:comm-generators}. 
\end{proof}

\noindent
Therefore, we can obtain a crude bound of $|K_i| \leq p^{O(d^3)}$ for any $i$. Also, the above lemma also gives an explicit description of the groups $K_S$. 

\begin{corollary}[Explicit description of $K_S$]\label{cor:exp-desc-K_S}
For any subset $S \subseteq [d]$, the group $K_S = \intersection_{i\in S} K_i$ consists of matrices $A = (A_{i,j})$ of the following form:
\begin{align*}
  A_{i,j} = \begin{cases}
    1 & \text{if }i = j,\\
    \text{a polynomial of degree $(j - i)\bmod{d}$} & \text{if }j\neq i \text{ and }\set{i,i+1,\ldots,j-1}_{\bmod{d}} \cap S = \emptyset,\\
    0 & \text{otherwise}.
  \end{cases}
\end{align*}
\end{corollary}

\noindent
Recall the \emph{other} set of subgroups defined for each $S \subseteq [d]$:
\[
  \widetilde{K_S} := \inangle{e_{i,i+1}(at+b)\;:\; a,b\in \F_p\,,\, i\notin S}.
\]
These groups can also be explicitly described.

\begin{lemma}[Explicit description of $\widetilde{K_S}$]\label{lem:exp-desc-K_S-tilde} For any $\emptyset \neq S \subseteq [d]$, the group  $\widetilde{K_S}$ is the set of all $d\times d$ matrices $A = (a_{ij})$ of the form
  \[
    a_{i,j} = \begin{cases}
      1 & \text{if }i = j,\\
      \text{a polynomial of degree $\leq j-i$} & \text{if }j\neq i \text{ and }\set{i,i+1,\ldots,j-1}_{\bmod{d}} \cap S = \emptyset,\\
      0 & \text{otherwise}.
    \end{cases}
  \]
\end{lemma}
\begin{proof}
Any $A\in \widetilde{K_S}$ can be expressed as  $A = B_1\cdots B_m$ where each $B_r = e_{i_{r},i_{r}+1}(\ell_r)$, for some linear polynomial $\ell_r$, with $i_r \notin S$.
Then,
\begin{align*}
  A_{i,j} & = \sum_{\substack{i_1,\ldots, i_{m+1}\\i_1 = i\;,\;i_{m+1} = j}} (B_1)_{i_1,i_2} (B_2)_{i_2,i_3} \cdots (B_m)_{i_m,i_{m+1}}. 
\end{align*}
From the structure of each $B_r$, any nonzero contribution from the RHS must involve either $i_{r+1} = i_{r}$, or $i_{r+1} = i_r + 1$ if $r \notin S$. This forces that the only entries of $A$ that are nonzero, besides the diagonal, are at $(i,j)$ with none of $\set{i,i+1,\ldots, j-1}$ in $S$.

In the case when $\set{i,i+1,\ldots, j-1} \cap S = \emptyset$, the above
argument also shows that the entry $A_{i,j}$ has degree at most $j-i$. Furthermore, using
\cref{obs:comm-generators}, we can easily see that $e_{i,j}(f) \in \widetilde{K_S}$ for an arbitrary polynomial $f(t)$ of degree at most $j-i$. From this, we can deduce that the structure of $\widetilde{K_S}$ is exactly as claimed. 
\end{proof}

\begin{proof}[Proof of \autoref{cor:intersection-property-KS}]
Follows immediately from \autoref{cor:exp-desc-K_S} and \autoref{lem:exp-desc-K_S-tilde}. 
\end{proof}

\noindent
From this point on, since the groups $K_S$ and $\widetilde{K_S}$ are identical, we drop the tilde notation and use $K_S$ for $\widetilde{K_S}$. 

\subsection{Connectivity of the coset complex}

\begin{lemma}\label{lem:KS-intersection-property}
  Let $S \subset [d]$ with $|S| \leq d-2$. Then,
  \[
    K_S= \inangle{K_{S} \cap K_i\colon i \in [d] \setminus
      S}. 
  \]
\end{lemma}
\begin{proof}
  It is clear that $K_S$ is a superset of the RHS. It only remains to show that the other containment also holds. To see this, consider an arbitrary generator $e_{j,j+1}(r)$ of $K_S$. Since $j \notin S$ and $|S| \leq d-2$, there is some $i \in [d] \setminus \inparen{S \union \set{j}}$. Therefore, $e_{j,j+1}(r) \in K_S \cap K_i$ and hence is generated by the RHS. 
\end{proof}

\noindent
Combining the above lemma with \cref{lem:conn-criterion} and \cref{lem:link-structure}, we have the following corollary.

\begin{corollarywp}
  For the coset complex $\mathcal{X}(G,\set{K_1,\ldots, K_d})$ defined by the above groups, the $1$-skeleton of every link is connected. 
\end{corollarywp}

\section{Spectral expansion of the complex}

In this section we prove that the coset complex $\X(G,\set{K_1,\ldots,
  K_d})$ is a good spectral HDX. The \autoref{thm:trickle-down} states that it
suffices to show that the $1$-dimensional links of faces in $\X(d-3)$ are good spectral expanders.

\subsection{Structure of $1$-dimensional links}

One dimensional links of the coset complex constructed are links of $v
\in \mathcal{X}(G,\set{K_1,\ldots, K_d})$ of size exactly $d-2$ (which are elements of $\X(d-3)$). Any
such $v$ can be written as $\set{gK_1,\ldots, gK_d} \setminus
\set{gK_{i},gK_{j}}$ for $i,j\in [d]$ with $i \neq j$ and $g \in
G$. Since the link of $v$ is isomorphic to the link of $\set{K_1,\ldots, K_d} \setminus
\set{K_{i},K_{j}}$, we might as well assume that $g=\mathrm{id}$.
These happen to be of two types depending on whether $i$ and $j$ are
consecutive or not. 

\begin{observation}
  Consider $v = \set{K_1,\ldots, K_d} \setminus \set{K_i,K_j}$ where $i$ and $j$ are \emph{not} consecutive (i.e. $(i-j) \neq \pm 1 \bmod d$). Then the $1$-dimensional link of $v$ is a complete bipartite graph.
\end{observation}
\begin{proof}
  Note that since $j \neq i\pm 1$, we have $[e_{i,i+1}(r_1), e_{j,j+1}(r_2)] = \mathrm{id}$ by \autoref{obs:comm-generators}. Hence, these two elements commute.

  The link of $v$ corresponds to the coset complex $\mathcal{X}(H,\set{H_1,H_2})$ where
  \begin{align*}
    H & = K_{[d]\setminus\set{i,j}} =  \inangle{e_{i,i+1}(at+b),e_{j,j+1}(at+b) \;:\; a,b\in \F_p},\\
    H_1 & = K_{[d]\setminus\set{i}} =  \inangle{e_{i,i+1}(at+b) \;:\; a,b\in \F_p},\\
    H_2 & = K_{[d]\setminus\set{j}} =  \inangle{e_{j,j+1}(at+b) \;:\; a,b\in \F_p}.
  \end{align*}
  Thus, the groups $H_1$ and $H_2$ commute with each other and hence any element of $h\in H$ can be written as $h = g_1 \cdot g_2$ where $g_1 \in H_1$ and $g_2 \in H_2$. \autoref{obs:nonempty-coset-intersection} implies that the resulting graph is the complete bipartite graph.    
\end{proof}

The interesting case is when $v = \set{K_1,\ldots, K_d}\setminus\set{K_i,K_{i+1}}$. 
Without loss of generality, we may focus on the link of $v = \set{K_3,K_4,\ldots, K_d}$. This corresponds to the coset complex $\mathcal{X}(H,\set{H_1, H_2})$
where
\begin{align*}
  H & = K_{3,4,\ldots, d} = \inangle{e_{1,2}(at + b), e_{2,3}(at + b) \;:\; a,b \in \F_p},\\
  H_1 & = K_{2,3,4,\ldots, d} = \inangle{e_{1,2}(at + b) \;:\; a,b \in \F_p},\\
  H_2 & = K_{1,3,4,\ldots, d} = \inangle{e_{2,3}(at + b) \;:\; a,b \in \F_p}.
\end{align*}
Hence, it suffices to focus on the first three rows and columns of these matrices as the rest of them are constant. Written down explicitly, 
\begin{align*}
H & = \setdef{\begin{bmatrix}
  1 & \ell_1 & Q\\
  0 & 1 & \ell_2\\
  0 & 0 & 1
\end{bmatrix}}{\begin{matrix}\ell_1,\ell_2\text{ are linear polynomials in $\F_p[t]$ }\\\text{and $Q$ is a quadratic polynomial in $\F_p[t]$}
\end{matrix}},\\
  H_1 & = \setdef{\begin{bmatrix}
  1 & \ell & 0\\
  0 & 1 & 0\\
  0 & 0 & 1
\end{bmatrix}}{\ell\text{ is a linear polynomial in $\F_p[t]$}},\\
    H_2 & = \setdef{\begin{bmatrix}
  1 & 0 & 0\\
  0 & 1 & \ell\\
  0 & 0 & 1
\end{bmatrix}}{\ell\text{ is a linear polynomial in $\F_p[t]$}}.
\end{align*}
Multiplication of an arbitrary element of $H$ with an arbitrary element of $H_1$ is of the form
\[
\begin{bmatrix}
1 & \ell_1 & Q\\
0 & 1 & \ell_2\\
0 & 0 & 1
\end{bmatrix} \cdot \begin{bmatrix}
  1 & \ell & 0\\
  0 & 1    & 0\\
  0 & 0    & 1
  \end{bmatrix} = \begin{bmatrix}
    1 & \ell_1 + \ell & Q\\
    0 & 1 & \ell_2\\
    0 & 0 & 1
    \end{bmatrix}.
\]
Note that the a unique choice of $\ell$ that makes the $(1,2)$-th entry of the RHS zero is $\ell = -\ell_1$. Thus, each coset of $H_1$ in $H$ has a unique representative of the form $M_1(\ell,Q)$ described below, and similarly, each coset of $H_2$ of $H$ has a unique representative of the form $M_2(\ell,Q)$. 
\[
M_1(\ell,Q) :=   \begin{bmatrix}
    1 & 0 & Q\\
    0 & 1 & \ell\\
    0 & 0 & 1
  \end{bmatrix} \quad,\quad M_2(\ell,Q) :=   \begin{bmatrix}
    1 & \ell & Q\\
    0 & 1 & 0\\
    0 & 0 & 1
  \end{bmatrix}
\]
respectively, 
where $\ell$ is a linear polynomial and $Q$ is a quadratic polynomial in $\F_p[t]$ since any arbitrary element of $H$ can be uniquely written as
\begin{align*}
  \begin{bmatrix}
    1 & \ell_1 & Q\\
    0 & 1 & \ell_2\\
    0 & 0 & 1
  \end{bmatrix} & = \begin{bmatrix}
    1 & \ell_1 & Q - \ell_1\ell_2\\
    0 & 1 & 0\\
    0 & 0 & 1
  \end{bmatrix}\begin{bmatrix}
    1 & 0 & 0\\
    0 & 1 & \ell_2\\
    0 & 0 & 1
  \end{bmatrix}\\
  &= \begin{bmatrix}
    1 & 0 & Q - \ell_1\ell_2\\
    0 & 1 & \ell_2\\
    0 & 0 & 1
  \end{bmatrix}
  \begin{bmatrix}
    1 & \ell_1 & 0\\
    0 & 1 & 0\\
    0 & 0 & 1
  \end{bmatrix}.
\end{align*}

\begin{lemma}
  For linear polynomials $\ell_1,\ell_2 \in \F_p[t]$ and quadratic polynomials $Q_1,Q_2 \in \F_p[t]$, we have that
  \[
    M_1(\ell_1,Q_1) H_1 \cap M_2(\ell_2,Q_2) H_2 \neq \emptyset \quad\Longleftrightarrow\quad \ell_1\ell_2 = Q_1 - Q_2.
  \]
\end{lemma}
\begin{proof}
  Note that matrices in $H_1H_2$ are of the form
  \[
    \begin{bmatrix}
      1 & \ell_1 & 0\\
      0  & 1 & 0\\
      0  & 0 & 1
    \end{bmatrix}
    \begin{bmatrix}
      1 & 0 & 0\\
      0  & 1 & \ell_2\\
      0  & 0 & 1
    \end{bmatrix} = \begin{bmatrix}
      1 & \ell_1 & \ell_1\ell_2\\
      0  & 1 & \ell_2\\
      0  & 0 & 1
    \end{bmatrix}.
  \]
  By \cref{obs:nonempty-coset-intersection}, the cosets have a non-empty intersection if and only if
  \[
    H_1 H_2 \ni M_1(\ell_1,Q_1)^{-1} M_2(\ell_2,Q_2) = \begin{bmatrix}
      1 & 0 & -Q_1\\
      0 & 1 & -\ell_1\\
      0 & 0 & 1
    \end{bmatrix}\begin{bmatrix}
      1 & \ell_2 & Q_2\\
      0 & 1 & 0\\
      0 & 0 & 1
    \end{bmatrix} = \begin{bmatrix}
      1 & \ell_2 & Q_2 - Q_1\\0 & 1 & -\ell_1 \\ 0 & 0 & 1
    \end{bmatrix}
  \]
  which happens if and only if $(\ell_2)(-\ell_1) = Q_2 - Q_1$ which is the same as $Q_1 - Q_2 = \ell_1\ell_2$. 
\end{proof}

Therefore, the $1$-dimensional link is the bipartite graph $A=(U,V,E)$
with left and right vertices identified by pairs $(\ell,Q)$ where
$\ell$ and $Q$ are linear and quadratic polynomials in $\F_p[t]$
respectively, with
$(\ell_1,Q_1) \sim (\ell_2,Q_2) \Leftrightarrow \ell_1\ell_2 = Q_1 +
Q_2$ (by associating $M_1(\ell,Q)$ with the tuple $(\ell,Q)$ on the
left, and $M_2(\ell,Q)$ with the tuple $(\ell,-Q)$ on the right).

Note that $A$ is an undirected, $p^2$-regular bipartite graph with  $p^5$ vertices on each side.
It suffices to show that $A$ is a good expander.

Kaufman and Oppenheim~\cite{KaufmanO2018} prove the expansion
properties of this graph using representation theory of the associated
groups, while we directly analyse the spectral gap of the adjacency
matrix associated with this graph. O'Donnell and Pratt~\cite[Case 2 in
the Proof of Theorem 3.23]{ODonnellP2022} give yet another proof of the spectral gap using the Polynomial Identity Lemma (also referred to as the Schwartz-Zippel lemma).

\subsection{A related graph}

The following graph is the ``lines-points'' or the ``affine plane'' graph used by Reingold, Vadhan and Wigderson~\cite{ReingoldVW2005} (as the \emph{base graph} in construction of constant-degree expanders, using the zig-zag product). 
Let $\F_q$ be a finite field.
Consider the bipartite graph $B_q = (U',V', E')$ defined as follows:
\begin{align*}
  U'=V' & = \F_q \times \F_q,&E' & = \setdef{\inparen{(a,b),(c,d)}}{ac = b+d}.
\end{align*}
Note that the graph $B_q$ is $q$-regular as for any vertex $a,b,c \in F_q$, there is a unique $d \in \F_q$ such that $ac = b+d$ and thus the vertex $(a,b)$ has exactly $q$ neighbours in $B_q$. 
\begin{lemma}
  The $q$-regular bipartite graph $B_q$ is a
  $\frac{1}{\sqrt{q}}$-one-sided-spectral expander. 
\end{lemma}
\begin{proof}
  Let $B_q^2$ denote the graph whose adjacency matrix is the square of the adjacency matrix of $B_q$. Restricted to the vertices in $U'$, it is easy to see that
  \[
    \text{Number of edges between $(a,b)$ and $(c,d)$} = \begin{cases} 1 & \text{if $a \neq c$,}\\
      q & \text{if $a = c$ and $b = d$,}\\
      0 & \text{otherwise.}
    \end{cases}
  \]
  Therefore, the adjacency matrix of $B_q^2$ (restricted to $U'$) can be written\footnote{In the equation, the notation $\otimes$ refers to the Kronecker product, or tensor product of matrices. It is well-known that, for square matrices $A$ and $B$, the set of eigenvalues of $A \times B$ is all products of the form $\lambda_i \cdot \nu_j$ where $\lambda_i$ is an eigenvalue of $A$ and $\nu_j$ is an eigenvalue of $B$.} (under a suitable order of listing vertices) as
  \[
    q I_{q^2} + (J_q - I_q)\otimes J_q \quad\text{(where $J_q$ is the $q\times q$ matrix of $1s$)}.
  \]
  By observing that $J_q$ has eigenvalue of $q$ with multiplicity $1$, and eigenvalue $0$ with multiplicity $(q-1)$, a simple calculation shows that $B_q^2$ has eigenvalue of $q^2$ with multiplicity $1$, eigenvalue $q$ with multiplicity $q(q-1)$ and eigenvalue $0$ with multiplicity $q-1$. 
  Hence the unnormalized second largest eigenvalue of $B_q^2$ is $q$ and hence we have that the normalized second largest eigenvalue of $B_q$ is $1/\sqrt{q}$. 
\end{proof}

\subsection{Relating the graph $B_q$ with $A$}

Set $q = p^3$ so that $\F_q = \frac{\F_p[y]}{\mu(y)}$ for some
irreducible polynomial $\mu(y)$ of degree exactly $3$. Therefore, each
element in $\F_q$ is expressible as $a_0 + a_1y + a_2 y^2$ for some
$a_0,a_1,a_2\in \F_p$. Thus, the graph $B_q=(U',V',E')$ defined above,
for this setting of $q = p^3$, is a $p^3$-regular bipartite graph with
$p^6$ vertices on either side.

Let $U'' = V'' = \setdef{(a_0 + a_1y , b_0 + b_1y + b_2 y^2)}{a_0,a_1,b_0,b_1,b_2 \in \F_p}$, which is a subset of $U'$ and $V'$, respectively, of size $p^5$ each.

\begin{observation}
  The induced subgraph of $B_q$ on  $U'',V''$ is exactly the graph
  $A=(U,V,E)$ described earlier. 
\end{observation}
\begin{proof}
  Note that $((\ell_1(y),Q_1(y)),(\ell_2(y),Q_2(y))) \in E'$ if and only if
  \[
    \ell_1(y)\cdot \ell_2(y) = Q_1(y) + Q_2(y) \bmod{\mu(y)}. 
  \]
  However, since the above equation has degree at most $2$, we have
  \[
    \ell_1(y)\cdot \ell_2(y) = Q_1(y) + Q_2(y) \quad\Leftrightarrow\quad \ell_1(y)\cdot \ell_2(y) = Q_1(y) + Q_2(y) \pmod{\mu(y)},
  \]
  and the first equation is exactly the adjacency condition of the graph $A$. Hence, the induced subgraph of $B_q$ on $U'',V''$ is indeed the graph $A$. 
\end{proof}

Normally, induced subgraphs of expanders need not even be connected. However, the following lemma shows that there are some instances where we may be able to give non-trivial bounds on $\lambda$.

\begin{lemma}\label{lem:induced-subgraph-eig-bound}
  Suppose $X$ is a $d$-regular, undirected graph that is an induced subgraph of a $D$-regular graph $Y$. Then,
  \[
    \lambda(X) \leq \frac{D \lambda(Y)}{d}. 
  \]
\end{lemma}
\begin{proof} The \autoref{thm:cf} characterization of the second largest
  eigenvalue tells us that \(\lambda(X) = \max_{\veca \perp \mathbf{1}_{|X|}} \frac{\veca^T G \veca}{d \cdot \veca^T \veca}\) where $G$ is the adjacency matrix of $X$.
 Consider an arbitrary $\veca \in \R^{|X|}$ such that $\veca \perp \mathbf{1}_{|X|} = 0$. Since $X$ is an induced subgraph of $Y$, the vector $\veca$ can be padded with zeroes to obtain a vector $\vecb_\veca \in \R^{|Y|}$ such that $\vecb_\veca \perp \mathbf{1}_{\abs{Y}}$. Therefore, if $A_X$ and $A_Y$ are the normalised adjacency matrices of $X$ and $Y$, we have
\[
   \lambda(X) = \max_{\veca \perp \mathbf{1}_{|X|}} \frac{\veca^T A_X
   \veca}{\veca^T \veca} =  \frac{D}{d}\cdot \max_{\veca \perp \mathbf{1}_{|X|}}
   \frac{\vecb_\veca^T A_Y \vecb_\veca}{\vecb_\veca^T \vecb_\veca}
   \leq \frac{D}{d}\cdot \max_{\vecb \perp \mathbf{1}_{|Y|}}
   \frac{\vecb^T A_Y \vecb}{\vecb^T \vecb} =
   \frac{D\lambda(Y)}{d} \enspace . \qedhere
 \]
\end{proof}

\begin{corollary}
  The graph $A(U,V,E)$ corresponding to the $1$-dimensional links of $\mathcal{X}(G,\set{K_1,\ldots, K_d})$ is a $\frac{1}{\sqrt{p}}$-one-sided-spectral expander. 
\end{corollary}
\begin{proof}
  The graph $B_{p^3}$ is a bipartite, $p^3$-regular graph with $\lambda(B_{p^3}) \leq \frac{1}{p^{3/2}}$ and $A(U,V,E)$ is a $p^2$-regular graph that is an induced subgraph of $B_{p^3}$. Hence, by \cref{lem:induced-subgraph-eig-bound},
  \[
    \lambda(A) \leq \frac{p^3 \cdot (1/p^{3/2})}{p^2} = \frac{1}{\sqrt{p}}. \qedhere
  \]
\end{proof}

\subsection*{The final expansion bounds}

From the corollary above, we obtain the following theorem of Kaufman and Oppenheim.

\begin{theorem}[\cite{KaufmanO2018}]\label{thm:KO-HDX} For $p > (d-2)^2$, 
  the $(d-1)$-dimensional coset complex $\X(G,\set{K_1,\ldots,K_d})$
  is a $\frac1{\sqrt{p}-(d-2)}$-one-sided-spectral HDX.
\end{theorem}

\begin{proof}
  Follows directly from \autoref{thm:trickle-down} that
  $\X(G,\set{K_1,\ldots,K_d})$ is a $\gamma$-one-sided-spectral HDX for
  \[
    \gamma \leq \frac{\sfrac{1}{\sqrt{p}}}{1 - (d-2)(\sfrac{1}{\sqrt{p}})} = \frac{1}{\sqrt{p} - (d-2)}.\qedhere
  \]
\end{proof}

\paragraph{Constructing two-sided-spectral HDXs and standard expanders:}

The $(d-1)$-dimensional coset complex $\X(G,\set{K_1,\ldots,K_d})$ is
not a two-sided-spectral HDX as the 1-skeletons of the links of the faces in
$\X(d-3)$ are bipartite. However, if we restrict attention to the $k$-skeleton of
$\X$ for some $k<d-1$ then we can bound the least eigenvalue using the
descent theorem for least eigenvalue
(\autoref{thm:trickle-down-one-level}\eqref{thm:trickle-down-negative}). This
is summarized in the following corollary. 

\begin{corollary}\label{cor:KO-exp} For $p > (d-2)^2$ and any $1 \leq
  k < d$
  the $k$-skeleton of the $(d-1)$-dimensional coset complex $\X(G,\set{K_1,\ldots,K_d})$
  is a $\max\set{\frac1{\sqrt{p}-(d-2)}\;,\;\frac{1}{d-k}}$-two-sided-spectral HDX.
  \end{corollary}

  In particular, if we set $k=1$ in the above corollary, we get a
  standard
  $\max\set{\frac1{\sqrt{p}-(d-2)}\;,\;\frac{1}{d-1}}$-two-sided-spectral expander.  This graph is a $d$-partite graph and hence its
  least eigenvalue is at most $\sfrac{-1}{(d-1)}$, while the above
  argument shows that it is least (and hence equal to)
  $\sfrac{-1}{(d-1)}$. 

 Thus, this not only yields an elementary construction and proof of
  one-sided-spectral HDXs (\cref{thm:KO-HDX}), but also one of
  standard spectral expander (\cref{cor:KO-exp}).

{\small
\bibliographystyle{prahladhurl}
\bibliography{HSa-bib}
}

\appendix

\section{Proof of the Descent Theorem}\label{sec:descent}

For the sake of completeness, we present the proof of
\autoref{thm:trickle-down} that asserts that proving spectral expansion
for the maximal faces is sufficient to obtain expansion of any
link. This exposition is essentially from the  lecture notes by Dikstein~\cite{Dikstein2019}. 

Let $(X,w)$ be a weighted $d$-dimensional simplicial complex. Let
$\mu_{d}=w|_{X(d)}$ be the distribution on the set $X(d)$ of
$(d+1)$-sized faces. This distribution induces distributions $\mu_i$ on
$X(i)$ in the natural way.\\ 

For two functions $f,g\colon X(0) \rightarrow \R$, define their \emph{inner product}
$\inangle{f,g}_X = \E_{u \sim \mu_0} [f(u) g(u)]$. We will drop the
subscript $X$ if it is clear from context. Note that, by the
definition of $\mu_1$, sampling $u$ according to $\mu_0$  can be
equivalently achieved  by sampling an edge $(u,v)$ according to $\mu_1$ and returning one of the points uniformly at random. Therefore,
\begin{equation}\label{eqn-innerprod1}
  \inangle{f,g}_X = \E_{u \sim \mu_0}[f(u) g(u)] = \E_{\set{u,v} \sim
    \mu_1}[f(u)g(u)] =\E_{v \sim \mu_0} \E_{u \sim X_v(0)}[f(u) g(u)] =\E_{v \sim \mu_0}[\inangle{f_v,g_v}_{X_v}],
\end{equation}
where $f_v,g_v\colon X_v(0) \rightarrow \R$ are the restrictions to the link of $v$. 

Define the \emph{adjacency operator} $A$ that, on a function $f\colon X(0) \rightarrow \R$ on vertices returns another function $Af$ on vertices defined via
\[
  Af(v) = \E_{u\sim v}[f(u)],
\]
where $u\sim v$ refers to a random neighbour of $v$ according to the distribution $u \sim \mu_0(X_v)$.
In other words, $A$ \emph{averages} $f$ over neighbours. Furthermore,
$A$ is self-adjoint with respect to the above inner product, i.e,
$\inangle{Af,g} = \inangle{f,Ag}$. Hence, it has $n$ real eigenvalues
and an orthonormal set of eigenvectors. Clearly
$A\mathbb{1} = \mathbb{1}$; the constant $1$ function is an
eigenvector for this operator (in fact, it is an eigenvector
corresponding to the largest eigenvalue 1). The remaining eigenvalues
are characterized by the \autoref{thm:cf}.

\begin{cftheorem}\label{thm:cf}
  Let $A\in \R^{n \times n}$ be an $n \times n$ matrix over the reals
  that is self-adjoint with respect to some inner product
  $\inangle{\cdot , \cdot }\colon \R^n \times \R^n \to \R$. Then $A$ has $n$ real
  eigenvalues $\lambda_1 \geq \cdots \geq \lambda_n$ which have the
  following characterization.
  \[
    \lambda_i = \max_{V \colon \dim V = i}\; \min_{0\neq  x \in V}
    \frac{\inangle{x,Ax}}{\inangle{x,x}} =  \min_{V \colon \dim V = n-i+1}\; \max_{ 0\neq x \in V}
    \frac{\inangle{x,Ax}}{\inangle{x,x}}.\qedhere
  \]
\end{cftheorem}

Similar to \autoref{eqn-innerprod1}, we have
\begin{align}
  \inangle{Af,g}_X & = \E_{\set{u,w} \sim \mu_1} [f(u) g(w)] = \E_{\set{u,v,w} \sim \mu_2} [f(u) g(w)]\nonumber\\
                 & = \E_{v \sim \mu_0}\insquare{\E_{\set{u,w}\sim  \mu_1(X_v)} [f(u) g(w)]}\nonumber\\
                 & = \E_{v \sim \mu_0}\insquare{\inangle{A_v f_v,g_v}_{X_v}}\label{eqn-innerprod2}
\end{align}
where $A_v$ denotes the adjacency operator restricted to the link $X_v$. \\

With the above notation, we can now state the theorem we wish to prove.
It suffices to prove the theorem in the case of $d = 3$ as we can obtain \autoref{thm:trickle-down} by induction.

\begin{theorem}\label{thm:trickle-down-one-level}
  Suppose $(X,w)$ is weighted $2$-dimensional simplicial complex. Then, we have the following two implications:
  \begin{enumerate}
  \item \label{thm:trickle-down-positive} Suppose the 1-skeleton of
    $X$ is connected and for every vertex $v \in X(0)$,
    $\inangle{A_v f,f} \leq \lambda \inangle{f,f}$ for all $f\colon
    X_v(0) \rightarrow \R$ with $f \perp \mathbb{1}_{X_v}$ for some
    $\lambda \in [0,1)$. Then, for
    any $g\colon X(0) \rightarrow \R$ with $g \perp \mathbb{1}_X$, we
have $\inangle{Ag,g} \leq \gamma \inangle{g,g}$ where $\gamma \leq
\frac{\lambda}{1- \lambda}$.
  \item \label{thm:trickle-down-negative} Suppose the 1-skeleton of
    $X$ is non-empty and for every vertex $v \in X(0)$, we have
    $\inangle{A_v f,f} \geq \eta \inangle{f,f}$ for all $f\colon
    X_v(0) \rightarrow \R$ for some $\eta \in [-1,1)$. Then, for any $g\colon X(0) \rightarrow \R$, we
have $\inangle{Ag,g} \geq \gamma \inangle{g,g}$ where $\gamma \geq
\frac{\eta}{1- \eta}$.
\end{enumerate}

\end{theorem}

\noindent Before we see a proof of this, let us  see how \autoref{thm:trickle-down} follows from this.

\begin{dtheorem}[\autoref{thm:trickle-down} restated]\label{thm:td-full}
  Suppose $(X,w)$ is a non-empty $d$-dimensional weighted simplicial complex with the following properties.
  \begin{itemize}\itemsep0pt
  \item The $1$-skeleton of every link is connected.
  \item For all $v \in X(d-2)$, the link $(X_v,w_v)$ is a
    $\lambda$-one-sided-spectral expander for some $\lambda < \frac1{d-1}$. I.e., there is a $\lambda > 0$ such that, for every $v \in X(d-2)$ and every $g\colon X_v(0) \rightarrow \R$ with $g \perp \mathbb{1}$, we have
  \[
\inangle{A_v g,g} \leq \lambda \inangle{g,g}.
  \]

  \end{itemize}
Then, $(X,w)$ is a $\gamma$-one-sided-spectral HDX for $\gamma \leq \frac{\lambda}{1 - (d-1)\lambda}$. That is, for any $v \in X(-1) \union \cdots \union X(d-2)$ and every $g\colon X_v(0) \rightarrow \R$ with $g \perp \mathbb{1}$, we have $\inangle{A_v g,g} \leq \gamma \inangle{g,g}$.\\

Furthermore, suppose we also know that there is a $\eta \in [-1,0)$ such that, for every $v \in X(d-2)$ and every $g\colon X_v(0) \rightarrow \R$, we have $\inangle{A_v g,g} \geq \eta \inangle{g,g}$. Then, $X$ is a $\gamma$-two-sided-spectral HDX with 
\[
  \gamma \leq \max\inparen{\frac{\lambda}{1 - (d-1)\lambda}, \abs{\frac{\eta}{1 - (d-1)\eta}}}.
\]
That is, for every $g \colon X_v(0) \rightarrow \mathbb{R}$ with $g
\perp \mathbb{1}$, we have $|\inangle{A_v g,g} | \leq \gamma \inangle{g,g}$.
\end{dtheorem}
\begin{proof}
  For any $i \leq d-2$, let
  \[\lambda_i = \min_{v\in X(i)} \max_{\substack{g:X_v(0) \rightarrow \R\\g\perp \mathbb{1}}}\frac{\inangle{A_v g,g}}{\inangle{g,g}},
  \]
  the smallest one-sided-spectral expansion with respect to $X(i)$. From repeated applications of  \autoref{thm:trickle-down-one-level},
  \[
    \lambda_{-1} \leq \frac{\lambda_0}{1 - \lambda_0} \leq \frac{\lambda_{1}/(1-\lambda_{1})}{1 - (\lambda_{1}/(1-\lambda_{1}))} = \frac{\lambda_{1}}{1 - 2\lambda_{1}} \leq \cdots \leq \frac{\lambda_{d-2}}{1 - (d-1)\lambda_{d-2}}
  \]
  which eventually completes the proof for one-sided-spectral expansion.

  For two-sided-spectral expansion, we also have to show that all the eigenvalues are bounded away from $-1$. One again, let $\eta_i$ be such that
  \[
    \eta_i = \max_{v\in X(i)} \min_{\substack{g\colon X_v(0) \rightarrow \R\\g\perp \mathbb{1}}}\frac{\inangle{A_v g,g}}{\inangle{g,g}}.    
  \]
  By repeated applications of \autoref{thm:trickle-down-one-level}~\eqref{thm:trickle-down-negative}, we obtain
  \[
    \eta_{-1} \geq \frac{\eta_0}{1 - \eta_0} \geq \frac{\eta_{1}/(1-\eta_{1})}{1 - (\eta_{1}/(1-\eta_{1}))} = \frac{\eta_{1}}{1 - 2\eta_{1}} \geq \cdots \geq \frac{\eta_{d-2}}{1 - (d-1)\eta_{d-2}}
  \]
  Together, we have that $X$ is a $\gamma$-two-sided-spectral HDX for
  \[
  \gamma = \max\inparen{\frac{\lambda}{1 - (d-1)\lambda}, \abs{\frac{\eta}{1 - (d-1)\eta}}}.\qedhere
  \]
  \end{proof}

\begin{proof}[Proof of \autoref{thm:trickle-down-one-level}]
  Let $g$ be an eigenvector that satisfies $\inangle{g,g} = 1$ and
  $g\perp \mathbb{1}_X$ that maximises (or minimises) $\inangle{Ag,g}$, and $\gamma =
  \inangle{Ag,g}$ be the extremal value. In particular, $Ag =
  \gamma\cdot g$. From \eqref{eqn-innerprod2} we have
    $\gamma = \inangle{Ag,g} = \E_v \insquare{\inangle{A_v g_v,g_v}}$.

    Even though $g \perp \mathbb{1}_X$, the \emph{local} component $g_v$ need not be perpendicular to $\mathbb{1}_{X_v}$. Hence,  let us write $g_v = \alpha_v \mathbb{1}_{X_v} + g_v^{\perp}$ where $g_v^{\perp} \perp \mathbb{1}_{X_v}$; we shall drop the subscript from $\mathbb{1}_{X_v}$ for the sake of brevity as the length of the vector will be clear from context. Note that $\alpha_v = \inangle{g_v, \mathbb{1}} = \E_{w\in X_v(0)} [g_v] = Ag(v)$. Therefore, $\E_v [\alpha_v^2] = \inangle{Ag,Ag} = \gamma^2$. Hence,
  \begin{align}
    \gamma = \inangle{Ag,g} & = \E_v \insquare{\inangle{A_v g_v,g_v}} = \E_v \insquare{\alpha_v^2 + \inangle{A_v g_v^\perp, g_v^{\perp}}}\label{eqn:trickle-LHS}
  \end{align}
  We shall now focus on the proof of \autoref{thm:trickle-down-one-level}~\eqref{thm:trickle-down-positive}. The other direction is exactly identical with the inequality flipped.

  In the case of \autoref{thm:trickle-down-one-level}~\eqref{thm:trickle-down-positive}, where we are given $\inangle{A_v g_v^\perp, g_v^{\perp}} \leq \lambda \inangle{g_v^\perp, g_v^{\perp}}$ for all $v\in X(0)$, we have
  \begin{align*}
    \gamma &= \E_v \insquare{\alpha_v^2 + \inangle{A_v g_v^\perp, g_v^{\perp}}} \leq \E_v \insquare{\alpha_v^2 + \lambda \inangle{g_v^\perp, g_v^{\perp}}} \\
    & = \E_v \insquare{(1-\lambda) \alpha_v^2 + \lambda \inangle{g_v, g_v}}\\
    &= (1-\lambda) \gamma^2 + \lambda.\\
    \implies \gamma(1-\gamma) & \leq \lambda (1- \gamma^2)\\
    \implies \gamma & \leq \lambda (1+\gamma) &\text{(connected, thus $\gamma < 1$)}\\
    \implies \gamma & \leq \frac{\lambda}{1 - \lambda}.
  \end{align*}

  In the case of \autoref{thm:trickle-down-one-level}~\eqref{thm:trickle-down-negative}, where we are given $\inangle{A_v g_v^\perp, g_v^{\perp}} \geq \eta \inangle{g_v^\perp, g_v^{\perp}}$ for all $v\in X(0)$, the same argument yields
  \begin{align*}
    \gamma = \E_v \insquare{\alpha_v^2 + \inangle{A_v g_v^\perp, g_v^{\perp}}}&\geq \E_v \insquare{\alpha_v^2 + \eta \inangle{g_v^\perp, g_v^{\perp}}} = (1 - \eta)\gamma^2 + \eta\\
    \implies \gamma & \geq \frac{\eta}{1 - \eta}\qedhere
  \end{align*}
\end{proof}

\section{Primer on Group Theory}\label{sec:group}

In this section, for completeness,  we shall note the basic definitions and properties of groups that are used in this exposition. 

\begin{definition}[Groups and subgroups]
  A set of elements $G$ equipped with a binary operation $\star: G \times G \rightarrow G$ is said to be a \emph{group} if it satisfies the following properties:
  \begin{description}
  \item[Associativity:] For all $g_1,g_2,g_3 \in G$, we have $\inparen{g_1 \star g_2} \star g_3 = g_1 \star \inparen{g_2 \star g_3}$.
  \item[Identity:] There exists an \emph{identity} element $\mathrm{id} \in G$ such that, for all $g\in G$, we have $g \star \mathrm{id} = \mathrm{id} \star g = g$.
  \item[Inverses:] For every element $g \in G$, there is an element $g^{-1} \in G$ such that $g \star g^{-1} = g^{-1} \star g = \mathrm{id}$. 
  \end{description}

  A subset $H \subseteq G$ is said to be a \emph{subgroup} of $G$ if
  $H$ the binary operation $\star$ restricted to $H$ satisfies  the
  above three properties (including the fact that $h_1 \star h_2 \in
  H$ for all $h_1,h_2 \in H$).

   \medskip

  Often the binary operation $\star$ is omitted and products just expressed as concatenation of elements.

\end{definition}

\begin{definition}[Cosets]
Given a subgroup $H$ of a group $G$, if $x \in G$ is an arbitrary element, the \emph{(left-)coset} of $H$ containing $x$, denoted by $x H$, is defined as the set
\[
x H = \setdef{x h}{h\in H}.
\]
Two cosets $x H$ and $y H$ are identical if and only if $x^{-1} y \in H$. Hence, any element $x'\in x H$ is also referred to as a \emph{coset representative} of $x H$ as $x' H = x H$. 
\end{definition}

Right-cosets are defined similarly. A subgroup $H$ is said to be
\emph{normal} if the right-coset and left-cosets agree for all $x$,
i.e., $xH = Hx, \forall x \in G$. 

Since two cosets of a subgroup $H$ of $G$ are either identical or
disjoint, the set of distinct cosets of a subgroup $H$ of $G$
partition the elements of $G$. If a subgroup $H$ is normal, this set
of cosets forms a group $G/H$, called the \emph{quotient group} of $H$
in $G$.\\

Suppose $H, K$ are subgroups of $G$, we will often consider the product $HK$ (or $H \star K$) which refers to the set $\setdef{h k}{h \in H\;,\; k\in K}$. 
It is worth stressing that $HK$ \emph{need not} be a subgroup of $G$ and the above just refers to a set of elements that can be expressed as an (ordered) product of an element in $H$ and an element in $K$. 

For an arbitrary set $S$ of $G$, we will define $\inangle{S}$ as the smallest subgroup of $G$ that contains the set $S$. This is also referred to as the \emph{group generated} by $S$.\\

In general, the binary operation $\star$ is order dependent. Groups where $g_1 g_2 = g_2 g_1$ for all $g_1,g_2 \in G$ are said to be \emph{commutative} or \emph{Abelian} groups. The following notion of \emph{commutators} (and \emph{commutator subgroups}) is a way to measure \emph{how non-commutative} a group $G$ is.

\begin{definition}[Commutators]\label{defn:commutator}
  For a pair of elements $g, h\in G$, we shall define the \emph{commutator} of $g,h$ (denoted by $[g,h]$) as
  \[
    [g,h] := g^{-1} h^{-1} g h. 
  \]
  The \emph{commutator subgroup} of $G$, denoted by $[G,G]$ is the \emph{group generated by all commutators}. That is,
  \[
    [G,G] := \inangle{\setdef{[g,h]}{g,h\in G}}.\qedhere
  \]
\end{definition}

Note that if $G$ is Abelian, then $[G,G] = \set{\mathrm{id}}$. As mentioned earlier, the commutator subgroup can be thought of as a way of describing how non-Abelian a group is. In fact, the commutator subgroup of $G$ is the smallest \emph{normal} subgroup $H$ of $G$ such that the \emph{quotient} $G/H$ is Abelian (although these are concepts that are not necessary to follow this exposition).

\end{document}